\documentclass[12pt,a4paper,reqno]{amsart}
\usepackage{amsfonts, amsmath, amssymb, amsthm, amsbsy}

\newtheorem{proposition}{Proposition}
\theoremstyle{definition}
\newtheorem{example}{Example}
\newtheorem{define}{Definition}
\theoremstyle{remark}
\newtheorem{rem}{Remark}

\newcommand{\pinner}{\mathbin{\mathchoice
   {\hbox{\vrule width0.6em depth0pt height0.4pt
   \vrule width0.4pt depth0pt height0.8ex}}
   {\hbox{\vrule width0.6em depth0pt height0.4pt
   \vrule width0.4pt depth0pt height0.8ex}}
   {\hbox{\kern0.14em
   \vrule width0.48em depth0pt height0.4pt
   \vrule width0.4pt depth0pt height0.6ex\kern0.14em}}
   {\hbox{\kern0.1em
   \vrule width0.39em depth0pt height0.4pt
   \vrule width0.4pt depth0pt height0.5ex\kern0.1em}}}}
\newcommand{\inner}{\pinner\,}

\newcommand{\sD}{\mathcal{D}}
\newcommand{\cD}{\mathcal{D}}

\DeclareMathOperator{\tr}{tr}
\DeclareMathOperator{\str}{str}
\DeclareMathOperator{\Mat}{Mat}
\DeclareMathOperator{\Ann}{Ann}
\DeclareMathOperator{\sdet}{sdet}

\newcommand{\tu}{\tilde{u}}
\newcommand{\txi}{\tilde{\xi}}
\newcommand{\tv}{\tilde{v}}

\newcommand{\cE}{\mathcal{E}}
\newcommand{\cEinf}{\mathcal{E}^{\infty}}

\newcommand{\cC}{\mathcal{C}}
\newcommand{\veps}{\varepsilon}
\newcommand{\BBR}{\mathbb{R}}

\newcommand{\BBC}{\mathbb{C}}
\newcommand{\BBZ}{\mathbb{Z}}
\newcommand{\bu}{\boldsymbol{u}}
\newcommand{\dd}{\partial}
\newcommand{\Id}{{\mathrm d}}
\newcommand{\fg}{\mathfrak{g}}
\newcommand{\bF}{\bar{F}}

\newcommand{\gl}{\mathfrak{gl}}
\newcommand{\gsl}{\mathfrak{sl}}

\newcommand{\by}[1]{\textit{{#1}}}
\newcommand{\jour}[1]{\textit{{#1}}}
\newcommand{\vol}[1]{\textbf{{#1}}}
\newcommand{\book}[1]{\textrm{{#1}}}

\newcommand{\bbl}{\underline{\boldsymbol{[}}}
\newcommand{\bbr}{\underline{\boldsymbol{]}}}

\newcommand{\bone}{\bar{1}}
\newcommand{\bzero}{\bar{0}}

\newcommand{\TD}[2][{}]{\tfrac{{\mathrm{d}}^{#1}}{\mathrm{d}{#2}^{#1}}}
\newcommand{\TDm}[2][{}]{\tfrac{{\mathrm{d}}^{|#1|}}{\mathrm{d}#2^{#1}}}
\newcommand{\half}{\textstyle{\frac{1}{2}}}

\title[Gardner's deformations of the graded KdV equations revisited]%
{Gardner's deformations of the graded\\ 
Korteweg\/--\/de Vries equations revisited}

\author[A.~V.~Kiselev]{A.~V.~Kiselev${}^{\dag}$}
\thanks{${}^{\dag}$
  \textit{Address}:
  Johann Bernoulli Institute for Mathematics and Computer Science,
  University of Groningen,
  P.O.Box~407, 9700AK Groningen, The Netherlands.\quad
  \textit{E-mail}: \texttt{A.V.Kiselev\symbol{"40}rug.nl}}

\author[A.~O.~Krutov]{A.~O.~Krutov$^{\ddag}$}  
\thanks{${}^{\ddag}$%
Corresponding author.\quad
  \textit{Address}: %
  Department of Higher Mathematics, Ivanovo State Power
  University, Rabfa\-kov\-skaya str.~34, Ivanovo, 153003 Russia.
\quad \textit{E-mail}:
\texttt{krutov\symbol{"40}math.ispu.ru}
}

\date{August 2, 2011}



\begin{document}
\begin{abstract}
We solve 
the Gardner deformation problem 
for the $N{=}2$ supersymmetric $a{=}4$ Korteweg\/--\/de Vries equation
(P.~Mathieu, 1988). We show that a known zero\/-\/curvature
representation for this superequation
yields the system of
new nonlocal variables such that their derivatives
contain the Gardner deformation for the classical KdV equation.
\end{abstract}
\maketitle


\subsection*{Introduction}
The classical problem of construction of the Gardner
deformation~\cite{Miura68} for an infinite\/-\/dimensional completely
integrable system of evolutionary partial differential equations
essentially amounts to finding a recurrence relation between the
integrals of motion. For the $N{=}2$ supersymmetric generalizations of
the Korteweg\/--\/de Vries equation~\cite{MathieuNew,MathieuOpen},
the deformation problem was posed when the integrable triplet of such
super\/-\/systems was discovered. Various attempts to solve it were
undertaken since then (e.g., see~\cite{MathieuNew}) but the progress
was limited. The first solution for the $N{=}2$, $a{=}4$ SKdV in the
triplet $a\in\{-2,1,4\}$ was achieved in~\cite{HKKW}; in that paper, we
stated the `no\/-\/go' theorem about the impossibility to deform
this super\/-\/equation in terms of the superfield (but not the
impossibility to deform it by treating the components of the
super\/-\/field separately, see below), c.f.~\cite{MathieuNew}. We then
presented the two\/-\/step solution of the deformation problem:
We obtained the Gardner deformation for the Kaup\/--\/Boussinesq
equation, which is the bosonic limit of the super\/-\/equation that
precedes the $N{=}2$, $a{=}4$ super\/-\/KdV in its hierarchy.
We thus derived the recurrence relation between the Hamiltonians of the
bosonic limit hierarchy and then we showed how each conserved density
is extended to the super\/-\/density for the $N{=}2$
super\/-\/system. In other words, we 
deformed the bosonic subsystem of the super\/-\/equation at hand within
the frames of the classical scheme~\cite{Miura68}, whence we recovered
the full $N{=}2$ supersymmetry\/-\/invariance.

In this paper we re\/-\/address, from a principally different
viewpoint, the Gardner deformation problem for a vast class of (not
necessarily supersymmetric) KdV\/-\/like systems. Namely,
in~\cite{HKKW} we emphasized the geometric likeness of the Gardner
deformations and zero\/-\/curvature representations, each of them
manifesting the integrability of nonlinear systems. Indeed, both
constructions generate infinite sequences of nontrivial integrals of
motion. We stress that, in general, the two constructions are
\textsl{not} equivalent, although they provide similar results. It is
precisely this correspondence which we study in this paper for the
graded KdV systems.

Developing further the approach of~\cite{RoelofsThesis},
we reformulate the
Gardner deformation problem for the graded extensions of the KdV
equation in terms of constructing parameter\/-\/dependent families
of new bosonic and fermionic variables. We require that the
`nonlocalities' possess two defining properties
(\cite{KK2000,TMPh2006}): on one hand, they should reproduce the
classical Gardner deformation from~\cite{Miura68} under the shrinking
of the $N{=}2$ super\/-\/equation back to the KdV equation. On the
other hand, we consider the nonlocalities that encode the
parameter\/-\/dependent zero\/-\/curvature representations for the
super\/-\/systems at hand. In this reformulation, we solve
P.~Mathieu's Open problem~2 of~\cite{MathieuOpen} for the $N{=}2$
supersymmetric $a{=}4$\/-\/KdV equation. However, our approach is
applicable to a much wider class of completely integrable
(super-)\/systems.

This paper is structured as follows. We first fix some notation and
compare the notions of Gardner's deformations and zero\/-\/curvature
representations by using their unifying description in terms of
nonlocal structures for partial differential equations~(PDE). In
section~\ref{SecGraded} we proceed with this correspondence for
$\BBZ_2$-\/graded systems of evolutionary PDE and solve the Gardner
deformation problem for the $N{=}2$, $a{=}4$ SKdV~\eqref{SKdV}.
The nature of the new solution is geometric and it presents an
alternative to the analytic two\/-\/step algorithm that works for
graded systems and which we reported earlier in~\cite{HKKW}.




\section{Preliminaries}\label{Notation}
\noindent%
In this section we fix some notation, which
follows~\cite{MarvanDIPS22001,Jimbo,BVV}.
Let $\varSigma^{n}$ be an $n$--dimensional manifold, $1\leqslant n<\infty$,
let $\pi\colon E^{m+n}\to\varSigma^{n}$ be a
vector bundle over $\varSigma^{n}$ of fiber
dimension $m$. Consider the jet space $J^{\infty}(\pi)$ of sections of
the vector bundle $\pi$.
The local coordinates on $J^{\infty}(\pi)$ are composed by the
coordinates $x^i$ on $\varSigma^{n}$, coordinates $u^j$ along the fibers
of $\pi$, and coordinates $u^j_{\sigma}$ along the fibers of
$J^{\infty}(\pi)\to \varSigma$; here $\sigma$ is a multi-index and
$u^j_{\varnothing}\equiv u^j$.
The commuting vector fields $\TDm[\sigma]{x} = \partial/\partial x_i +
\sum_{k,\sigma}u^k_{\sigma i}\partial/\partial u^k_{\sigma}$ on
$J^{\infty}(\pi)$ are called \textsl{total derivatives}.
The operators $\TDm[\sigma]{x}$ that act on the space of smooth
functions on $J^{\infty}(\pi)$ are defined inductively by the formula
$\TDm[\sigma i]{x} = \TDm[\sigma]{x}\circ \TD{x^i}$.

Consider a system $\cE$ of $r$ parital differential equations,
\[
\cE = \left\{ F^{\ell}(x_i, u, \dots, u_{I}^k, \dots) = 0,\quad
  \ell=1,\dots,r \right\}.
\]
The system $\{F^\ell=0\}$ and all its differential consequences
$\TDm[\sigma]{x}F^{\ell}=0$, $|\sigma|\geq1$ generate the
infinite prolongation of~$\cE$, which we denote by~$\cEinf$.
The restrictions of $\TD{x^i}$ on $\cEinf$ determine the Cartan
distribution $\cC$ on the tangent space~$T\cEinf$.
Here and in what follows, the notation $\TD[x]{x}$ stands for the
restrictions of the total derivatives onto~$\cEinf$.
There is the decomposition of the tangent space $T\cEinf$ to the
direct sum of the horizontal (the Cartan distribution) and the vertical
vector spaces, $T\cEinf = \cC \oplus V\cEinf$.
Let $\Lambda^{1,0}(\cEinf)= \Ann \cC$ and
$\Lambda^{0,1}(\cEinf) = \Ann V\cEinf$ be
the $C^\infty(\cEinf)$-\/modules of contact $1$-\/forms and
horizontal $1$-\/forms vanishing on \(\cC\) and \(V\cEinf\),
respectively.
Let $\Lambda^r(\cEinf)$ denote the $C^\infty$--module 
of $r$-forms on  $\cEinf$. We have the decomposition
$\Lambda^r(\cEinf) = \bigoplus_{q+p=r}\Lambda^{p,q}(\cEinf)$, where
$\Lambda^{p,q}(\cEinf) =
\bigwedge^p\Lambda^{1,0}(\cEinf)\wedge\bigwedge^q\Lambda^{0,1}(\cEinf)$.
According to this decomposition, the exterior differential splits
to the sum  $d = \bar{d} + d_{\cC}$ of the horizontal differential $\bar{d}:
\Lambda^{p,q}(\cEinf)\to\Lambda^{p,q+1}(\cEinf)$ and the vertical
differential $d_{\cC}:\Lambda^{p,q}(\cEinf) \to
\Lambda^{p+1,q}(\cEinf)$. 

The differentials $d_{\cC}$ and $\bar{d}$ can be expressed in
coordinates by their actions on the elements
$\phi\in C^\infty(\cEinf)=\Lambda^{0,0}(\cEinf)$, whence
\[
 \bar{d}\phi  = \sum\limits_i \TD{x^i} \phi\, dx^i, \quad
d_{\cC} \phi = \sum\limits_{\sigma,k}
   \frac{\partial \phi}{\partial u^k_{\sigma}}\omega^k_{\sigma},
\]
here $\omega^k_{\sigma} =
du^k_{\sigma} - \sum\limits_i u^k_{\sigma i}\,dx^i$ are the Cartan forms.

\subsection{Zero curvature representations and coverings}
Consider the tensor product over $\BBC$ of the exterior algebra
$\bar{\Lambda}(\cEinf)=\bigwedge^*\Lambda^{i,0}(\cEinf)$
with a finite\/-\/dimensional complex
Lie algebra~$\fg$. The product is endowed with
the bracket  $[A\mu, B\nu] = [A,B]\mu\wedge\nu$ for
$\mu,\nu\in\bar{\Lambda}(\cEinf)$ and $A,B\in \fg$.
The tensor product $\bar{\Lambda}(\cEinf)\otimes\fg(\BBC)$ is a
differential graded associative
algebra with respect to the multiplication $A\mu\cdot B\nu = (A\cdot
B)\mu\wedge\nu$ induced by the ordinary matrix multiplication so that
\[
\bar{d}(\rho\cdot\sigma) =
   \bar{d}\rho\cdot\sigma + (-1)^r\rho\cdot\bar{d}\sigma
\]
for $\rho\in\bar{\Lambda}^r(\cEinf)\otimes\fg(\BBC)$ and
$\sigma\in\bar{\Lambda}^s(\cEinf)\otimes\fg(\BBC)$,
whereas
\[
[\rho,\sigma] = \rho\cdot\sigma - (-1)^{rs}\rho\cdot\sigma.
\]
The elements of $C^\infty(\cEinf)\otimes\fg$ will be called
the \textsl{$\fg$-matrices}~\cite{MarvanDIPS22001}.

\begin{define}[\cite{MarvanDIPS22001}]
A horizontal $1$-\/form
$\alpha\in\bar{\Lambda}^1(\cEinf)\otimes\fg$ is
called the $\fg$-valued \textsl{zero\/-\/curvature representation} for
the equation~$\cE$ if the Maurer\/--\/Cartan condition holds:
\begin{equation}\label{zcrf}
\bar{d}\alpha = \frac12[\alpha,\alpha].
\end{equation}
If $\alpha = \sum_{i=1}^n A_i\,\mathrm{d}x^i$,
where $A_i\in\fg$, is a $\fg$-\/valued
zero\/-\/curvature representation for $\cE$, then we have
\begin{multline*}
  0 = \bar{d}\alpha - \frac12[\alpha,\alpha] 
= \sum_{i < j} (\TD{x^i}
 A_j - \TD{x^j} A_i) dx^i\wedge dx^j - \sum_{i < j} [A_i, A_j] dx^i\wedge
  dx^j = 0.
\end{multline*}
Therefore, equation~\eqref{zcrf} is equivalent to the following set of
conditions upon the $\fg$-\/matrices~$A_i$:
\begin{equation}\label{zcrm}
  \TD{x^j} A_i - \TD{x^i} A_j + [A_i, A_j] = 0, \quad \forall i,j=1,\dots,m: i\neq j.
\end{equation}
\end{define}


The most interesting zero\/-\/curvature representations for~$\cE$ are
those which contain a non\/-\/removable spectral parameter;
in this case the system~$\cE$ is integrable.
(The parameter is \textsl{removable} if one obtains
gauge\/-\/equivalent zero\/-\/curvature representations,
see section~\ref{secgauge} for details,
at different values of the parameter;
otherwise, the parameter is non\/-\/removable, or ``essential.'')

We recall that $n$ is the dimension of the base $\varSigma^n$ for the
vector bundle $\pi$. From now on, we consider mainly
$(1+1)$-\/dimensional 
systems, \textit{i.e.}, we let $n=2$
and interpret one independent variable as the time~$t$ and the other as
the spatial coordinate~$x$.
With the conventions 
$n=2$, $x^1=x$, $x^2=t$, $A_1=A$, and $A_2=B$,
the Maurer\/--\/Cartan equations (\ref{zcrf}--\ref{zcrm}) read
\begin{equation}
\label{zcrmprime}
\TD{t} A - \TD{x} B + [A, B] = 0. \tag{\ref{zcrm}$'$}
\end{equation}
This 
is the compatibility condition for the auxiliary linear system
\[
\Psi_x = A \Psi,  \quad \Psi_t = B \Psi,
\]
where 
$\Psi$ is a wave function and
matrices $A$ and $B$ belong to the tensor product
of a matrix Lie algebra $\fg$ and the ring $C^{\infty}(\cEinf)$ of
smooth functions on the prolongation $\cEinf$.
If 
the matrices $A$ and $B$ depend on the
spectral parameter and it is non\/-\/removable, then
the equation $\cE$ is integrable~\cite{Faddeev}.

The construction of Gardner's deformations~\cite{Miura68} is another
way to prove the integrability of evolution equations~$\cE$.

\begin{define}[Gardner's deformation~\cite{Miura68}]
Let $\mathcal{E}=\left\{u_t=f(x,u,u_x,\ldots,u_k)\right\}$
be a system of evolutionary partial differential equations
upon the unknowns $u(x,t)$ in two variables.
Let $\mathcal{E}_{\varepsilon}
= \{\tilde{u}_t = f_{\varepsilon}(x,\tilde{u},\tilde{u}_{x},\ldots,\tilde{u}_k;\varepsilon) \mid
{f_{\varepsilon}\in\mathrm{im}\,\TD{x}}\}$
be a deformation of $\mathcal{E}$ such that at each point
$\varepsilon\in\mathcal{I}\subseteq\mathbb{R}$ there exists
the Miura contraction
$\mathfrak{m}_{\varepsilon} = \{ u =
u(\tilde{u},\tilde{u}_x,\ldots; 
\varepsilon\}\colon\cE_{\varepsilon}\to\cE$. 
Then the pair $(\mathcal{E}_{\varepsilon}, \mathfrak{m}_{\varepsilon})$
is the \textsl{Gardner deformation} for the system~$\mathcal{E}$.
\end{define}

One obtains the recurrence relations between the conserved densities
$\tu_n(x,u,u_x,\ldots)$ for $\cE$ using the contraction $\mathfrak{m}_{\veps}$
and the expansion $\tu=\sum_{n=0}^{+\infty}\tu_n\veps^n$ of the deformed
unknowns $\tu$ in~$\veps$.

In the recent paper~\cite{TMPh2006} we 
understood Gardner's deformations in the extended sense,
namely, in terms of coverings over PDE and diagrams of coverings.
The zero\/-\/curvature representations and Gardner's deformations can
be considered as such geometric structures\footnote{B\"ack\-lund
(auto)\/transformations between PDE appear in 
the same context. In~\cite{TMPh2006} we argued that the former, when
regarded as the diagrams, are dual to the diagram description of
Gardner's deformations.}
that obey some extra conditions. 

\begin{define}[\cite{BVV}]
Let the assumptions on p.~\pageref{Notation}
hold and let the notation be as previously fixed.
A \textsl{covering} over the equation $\cE$ is another (usually,
larger) system of partial differential equations $\tilde{\cE}$
endowed with the $n$-\/dimensional Cartan distribution $\tilde{\cC}$
and such that there is a mapping $\tau\colon\tilde{\cE}\to\cEinf$
for which, at each point $\theta\in\tilde{\cE}$
the tangent map $\tau_{*,\theta}$
is an isomorphism of the plane $\tilde{\cC}_{\theta}$ to the Cartan
plane $\cC_{\tau(\theta)}$ at the point~$\tau(\theta)$ in $\cEinf$.
\end{define}

In practice, the construction of a covering over $\cE$ means the
introduction of new nonlocal variables such that their compatibility
conditions lie inside the initial system~$\cEinf$.
Whenever the covering is indeed realized as the fibre bundle
$\tau\colon\tilde{\cE}\to\cE$, the forgetful map~$\tau$ discards the
nonlocalities.

In these tems, the zero\/-\/curvature representations and Gardner's
deformations are coverings of special kinds
(see Examples~\ref{exCovAndZcr} and~\ref{exGardnerCov} below).
In this paper we use the geometric closedness of the two notions and
construct new Gardner's deformations from known zero\/-\/curvature
representations (but this is \textsl{not always} possible).

\begin{example}[A zero\/-\/curvature representation for
the KdV equation]\label{exLaxZCR}
Consider the Korteweg\/--\/de Vries (KdV) equation~\cite{Miura68}
\begin{equation}
\label{kdv}
  \cE_{\text{KdV}} = \left\{ u_t + u_{xxx} + 6uu_x = 0 \right\}
\end{equation}
and its Lax representation~\cite{Miura68,Jimbo,Faddeev}
\[
  \mathcal{L}_t = [\mathcal{L}, \mathcal{A}],
\]
where
\begin{equation}
\label{LaxL}
  \mathcal{L} = \TD[2]{x} + u, \quad
  \mathcal{A} = -4\TD[3]{x} -6u\TD{x} -3u_x.
\end{equation}

The linear auxiliary problem~\cite{ZSh} is
\begin{align*}
  \psi_{xx} + u\psi &= \lambda\psi, \\
  -4\psi_{xxx}  -6u\psi_x -3u_x\psi &= \psi_t,
\end{align*}
By definition, put \(\psi_0 = \psi\) and \(\psi_1 = \psi_x\). We obtain
\begin{align*}
  \psi_{0;x} & = \psi_1,\\
  \psi_{1;x} & = (\lambda-u)\psi_0,\\
  \psi_{0;t} & = -4\TD{x}((\lambda - u)\psi_{0}) -6u\psi_1 -3 u_{x}\psi_0  = u_x\psi_0 + (-4\lambda - 2u)\psi_!,\\
  \psi_{1;t} & = (-4\lambda^2 +2u\lambda + 2u^2 + u_x)\psi_0  + (-u_x)\psi_1.
\end{align*}
We finally rewrite this system as two matrix equations~\cite{ZSh},
\[
\underbrace{\begin{pmatrix}
\psi_{0;x} \\ \psi_{1;x}
\end{pmatrix} }_{\psi_x}
=
\underbrace{\begin{pmatrix}
  0 & 1 \\
  \lambda - u & 0
\end{pmatrix}}_{A}\underbrace{\begin{pmatrix}
\psi_{0} \\ \psi_{1}
\end{pmatrix}}_{\psi}
\quad
\underbrace{\begin{pmatrix}
\psi_{0;t} \\ \psi_{1;t}
\end{pmatrix} }_{\psi_t}
=
\underbrace{\begin{pmatrix}
  u_{x}  & -4\lambda - 2u\\
  -4\lambda^2 + 2u\lambda + 2u^2 + u_{xx} & -u_{x}
\end{pmatrix}}_{B}\underbrace{\begin{pmatrix}
\psi_{0} \\ \psi_{1}
\end{pmatrix}}_{\psi}
\]
This yields an $\gsl_2(\BBC)$-\/valued zero\/-\/curvature
representation $\alpha^{\text{KdV}} = Adx + Bdt$ for the KdV
equation~\eqref{kdv}. 
The representation~$\alpha^{\text{KdV}}$ was rediscovered
in~\cite{MarvanSL2}.
\end{example}

\begin{example}[Zero\/-\/curvature representations as
coverings]\label{exCovAndZcr}
Let $\fg\mathrel{{:}{=}}\mathfrak{sl}_2(\BBC)$ as above.
We introduce the
standard basis $e,h,f$ in  $\fg$  such that
\[
 [e,h] = -2e,\quad [e,f] = h,\quad [f,h] = 2f.
\]
We consider, simultaneously, the matrix representation
\[
\rho:\ \mathfrak{sl}_2(\BBC)\to \{M \in \Mat(2,2) | \tr M{=}0\}
\]
of $\fg$ and its representation in the space of vector fields with
polynomial coefficients on the complex line with the coordinate~$w$:
\[
\begin{array}{rclrclrcl}
\rho(e)& = &\begin{pmatrix} 0 & 1 \\ 0 & 0 \end{pmatrix}, &
\rho(h)& = &\begin{pmatrix} 1 & 0 \\ 0 & -1 \end{pmatrix},&
\rho(f)& = &\begin{pmatrix} 0 & 0 \\ 1 & 0 \end{pmatrix},\\
\varrho(e)& = &1\cdot\partial/\partial w, \quad &
\varrho(h)& =&-2w\cdot\partial/\partial w, \quad &
\varrho(f) &=&-w^2\cdot\partial/\partial w.
\end{array}
\]
Let us decompose the matrices $A_i$ (which occur in the zero\/-\/curvature
representation $\alpha = \sum_iA_idx^i$)
with respect to the basis in the space~$\rho(\fg)$,
\begin{align}\label{atensmat}
A_i & = a_e^{(i)} \otimes \rho(e) + a_h^{(i)} \otimes \rho(h) + a_f^{(i)} \otimes \rho(f),
\end{align}
for $a^{(i)}_j \in C^{\infty}(\cEinf)$.

To construct the covering $\tilde{\cE}$ over $\cEinf$ with a new
nonlocal variable $w$, we switch from the representation $\rho$ to
$\varrho$. We thus obtain the vector fields
\begin{equation}
\label{atensvf}\tag{\ref{atensmat}$'$}
V_{A_i} =  a_e^{(i)} \otimes \varrho(e) + a_h^{(i)} \otimes \varrho(h) + a_f^{(i)} \otimes \varrho(f)
\end{equation}
such that the prolongations of the total derivatives
$\TD[i]{x}$ to $\tilde{\cE}$ are defined by the formula
\begin{equation}
\label{tilded}
  \tilde{\TD{x^i}} = \TD{x^i} - V_{A_i}.
\end{equation}
The extended derivatives act on the nonlocal variable $w$ as follows
\[
\tilde{\TD{x^i}} w = dw \inner ( - V_{A_i}).
\]
\end{example}

\begin{rem}
The commutativity of the prolonged total derivatives
${\bigl[\tilde{\TD{x^i}}, \tilde{\TD{x^j}}\bigr] = 0}$ with $i\neq j$
is equivalent to the Maurer\/--\/Cartan equation~\eqref{zcrm}:
Indeed, we have that
\begin{multline*}
0 = [\tilde{\TD{x^i}}, \tilde{\TD{x^j}}] = [\TD{x^i} - V_{A_i}, \TD{x^j} -
V_{A_j}] = [\TD{x^i}, \TD{x^j}]  - [\TD{x^i}, V_{A_j}] - [V_{A_j},
\TD{x^j}] + [V_{A_i}, V_{A_j}]= \\
{}= - V_{\TD{x^i} A_i} + V_{\TD{x^j} A_i} + V_{[A_i, A_j]} = V_{ \TD{x^j} A_i - \TD{x^i}
  A_j + [A_i, A_j]}\   \Leftrightarrow{} \  \TD{x^j} A_i - \TD{x^i}
  A_j + [A_i, A_j]=0.
\end{multline*}
This motivates the choice of the minus sign in~\eqref{tilded}.
\end{rem}

\begin{example}[A one\/-\/dimensional covering over the KdV equation]
One obtains the covering over the KdV equation from the
zero\/-\/curvature representation $\alpha$ (see Example~\ref{exLaxZCR})
by using representation~\eqref{atensvf} in the space of vector fields.
Applying~\eqref{atensvf} to the matrices $A$,\ $B\in\gsl_2(\BBC)$,
we construct the following vector fields with the nonlocal variable $w$:
\begin{align*}
V_A & = (1  - (\lambda -u)w^2)\cdot\partial/\partial w,\\
V_B & = \left[(-4\lambda - 2u)  - 2uw  -  (-4\lambda^2 + 2u\lambda + 2u^2 + u_{xx})w^2\right]\cdot\partial/\partial w.
\end{align*}
The prolongations of the total derivatives act on $w$ by the rules
\begin{subequations}\label{mcov}
\begin{align}
w_x & =  - 1  + (\lambda -u)w^2,\\
w_t & = - \left((-4\lambda - 2u) - 2u_xw  - (-4\lambda^2 + 2u\lambda +
  2u^2 + u_{xx})w^2\right).
\end{align}
\end{subequations}
We thus obtain the one\/-\/dimensional covering over the KdV
equation~\eqref{kdv}.
It depends on the non\/-\/removable~\cite{BVV} spectral
parameter~$\lambda$. In what follows we show that this covering is
equivalent to the covering~\eqref{miuracov} which is derived from
Gardner's deformation~\eqref{kdvdef} of the KdV equation~\eqref{kdv}.
\end{example}

\subsection{The projective substitution and nonlinear representations
of Lie algebras in the spaces of vector fields}
Suppose $\fg$ is a finite\/-\/dimensional Lie algebra. We shall use
the projective substitution~\cite{RoelofsThesis} to construct
a covering over the equation $\cE$ starting from a $\fg$-\/valued
zero\/-\/curvature representation for~$\cE$.

Let $M$ be an $m$-\/dimensional manifold with local coordinates
\[
v = (v^1, v^1, \dots, v^m)\in M, \ \text{and put}\
\partial_v = (\partial_{v^1}, \partial_{v^2}, \dots, \partial_{v^m})^t.
\]
For any $g\in\gl\subseteq\gl_n(\BBC)$, its representation
$V_g$ in the space of vector fields on~$M$ is given by the formula
\[
 V_g = v^tg\partial_v.
\]
We note that $V_g$ is linear in $v^i$. By construction, the
representation preserves the commutation relations in the initial Lie
algebra $\fg$:
\[
 [V_g, V_f] = [v^tg\partial_v, v^tf\partial_v] = v^t[g,f]\partial_v =
 V_{[g,f]}, \quad f, g\in\fg.
\]
At all points of $M$ where $v_1\neq0$ we consider the projection
\begin{equation}\label{projsub}
\pi : v^i \to w^i = \mu v^i/ v^1, \quad \mu\in\BBR
\end{equation}
and its differential \(d\pi: \partial_v \to \partial_{w}\).
The transformation $\pi$ yields the new coordinates on the open subset
of $M$ where $v^1\neq 0$ and corresponding subset of $TM$:
\[ w = (\mu, w^2, \dots, w^m), \quad \partial_{w} = (-\frac{1}{\mu}\sum\limits_{i=2}^m w^i\partial_{w^i}, \partial_{w^2}, \dots, \partial_{w^m}).
\]

Consider the vector field $W_g=d\pi(V_g)$. In coordinates, we have
\begin{equation}
\label{wroelofs}
W_g = w g \partial_{w}^{t}.
\end{equation}
We note that, generally, $W_g$ is nonlinear with respect to~$w^i$.
The commutation relations between the vector fields of such type
are also inherited from the relations in the Lie algebra $\fg$:
\[
  [W_g, W_f] = [d\pi(V_g), d\pi(V_f)] = d\pi([g,f]) = d\pi(V_{[g,f]}) = W_{[g,f]}.
\]

Using representation~\eqref{wroelofs} for the matrices $A$ and $B$
that determine the zero\/-\/curvature representation
$\alpha^{\text{KdV}}=Adx+Bdt$ for the KdV equation, we obtain their
realizations 
in terms of the vector fields:
\begin{align*}
W_A & = \frac{1}{\mu}( -\lambda w^2 + \mu^2 + u w^2)\,\partial/\partial w,\\
W_B & = \frac{1}{\mu}( -u_{xx} w^2 - 2u_x\mu w + 4\lambda^2w^2 - 4\lambda\mu^2 - 2\lambda u w^2 - 2\mu^2u - 2u^2 w^2)\,\partial/\partial w.
\end{align*}
Therefore, the prolongations of the total derivatives act on the
nonlocality $w$ as follows:
\begin{subequations}\label{rcov}
\begin{align}
w_x & = -\frac{1}{\mu}( -\lambda w^2 + \mu^2 + u w^2),\\
w_t & = -\frac{1}{\mu}( -u_{xx} w^2 - 2u_x\mu w + 4\lambda^2w^2 - 4\lambda\mu^2 - 2\lambda u w^2 - 2\mu^2u - 2u^2w^2).
\end{align}
\end{subequations}
The parameter $\mu$ is removable by the transformation $w\to\mu w$,
which rescales it to unit.
Applying this transformation to~\eqref{rcov}, we reproduce the
covering~\eqref{mcov}.

\begin{example}[A covering which is based on Gardner's
deformation]\label{exGardnerCov}
Consider the Gardner deformation~\cite{Miura68} of the KdV
equation~\eqref{kdv},
\begin{subequations}\label{kdvdef}
\begin{align}
\cE_{\veps} = & \left\{\tu_t  = - (\tu_{xx} + 3\tu^2 -
  2\veps^2\tu^3)_x \right\}, \label{kdvext}
\\
\mathfrak{m}_{\veps} = & \left\{ u  = \tu - \veps \tu_x - \veps^2\tu^2
  \right\}\colon\cE_\veps\to\cE_0.\label{miura}
\end{align}
\end{subequations}
Expressing $\tu_x$ from~\eqref{miura} and substituting it
in~\eqref{kdvext}, we obtain the one\/-\/dimensional covering over the
KdV equation,
\begin{subequations}
\label{miuracov}
\begin{align}
\tu_x& = \frac{1}{\veps}( \tu - u) - \veps\tu^2\label{miuracovx},\\
\tu_t& = \frac{1}{\veps}( u_{xx} + 2u^2) + \frac{1}{\veps^2} u_x +
\frac{1}{\veps^3}u +  \left(-2u_x - \frac{2}{\veps}u - \frac{1}{\veps^3}\right)\tu
+ \left(2\veps u + \frac{1}{\veps}\right)\tu^2, \label{miuracovext}
\end{align}
\end{subequations}

We claim that covering~\eqref{miuracov} is equivalent to the covering
that was obtained in~\cite[p.~277]{BVV} for the KdV equation.
To prove this, we first put $ \tu = - \tv / \veps$. We have
\begin{align*}
-\frac{\tv}{\veps} &= -\frac{1}{\veps^2}\tv - \frac{1}{\veps}u  - \frac{1}{\veps}\tv^2,\\
\intertext{in other words}
\tv_x & = u + \left(\tv + \frac{1}{2\veps}\right)^2 - \frac{1}{4\veps^2}.
\end{align*}
Next, we put $ p = \tv + 1/(2\veps)$, whence we obtain
\begin{subequations}
\label{covBVV}
\begin{align}
p_x &= u + p^2  - \frac{1}{4\veps^2},\\
p_t &= -u_{xx} - 2u^2 - \frac{1}{2\veps^2}u + \frac{1}{4\veps^4} -
2u_xp - (\veps^2u + \frac{1}{2})p^2.
\end{align}
\end{subequations}
Dividing~\eqref{mcov} by $w^2$, we conclude that
\begin{align*}
w_x &= - 1 + (\lambda-u)w^2,\\
\frac{w_x}{w^2} &= - \frac{1}{w^2} - u + \lambda.
\end{align*}
On the other hand, we put $p = 1/w$, whence $p_x = - w_x/w^2$, and set
$\lambda = 1/(4\veps^2)$. This brings~\eqref{mcov} to the same
notation as in formulas~\eqref{covBVV},
\begin{align*}
p_x &= u + p^2  - \frac{1}{4\veps^2},\\
p_x & = u + p^2 - \lambda.
\end{align*}
The corresponding one\/-\/form of the zero\/-\/curvature representation for
the KdV equation is equal to
\begin{equation}
\label{zcrBVV}
\alpha^{\text{KdV}}_{2} = \begin{pmatrix}
0 & \lambda - u\\
1 & 0
\end{pmatrix}dx
+
 \begin{pmatrix}
- u_{x} & - 4\lambda^2 + 2\lambda u  + 2u^2 +
  u_{xx}  \\
 - 4\lambda - 2u  &  u_{x}
\end{pmatrix} dt.
\end{equation}
In the next section we show that this zero curvature representation is
also equivalent to $\alpha^{\text{KdV}}$ from Example~\ref{exLaxZCR}.
\end{example}

\subsection{Gauge transformations}\label{secgauge}
Let $G$ be the Lie group of the Lie algebra $\fg$ (so that
$G=SL_2(\BBC)$ in the previous example).
Given an equation $\cE$, for any\/-\/zero curvature representation
$\alpha$ there exists the 
zero\/-\/curvature representation $\alpha^{S}$ such that
\begin{equation}
  \alpha^{S} = \bar{d}S\cdot S^{-1} + S\cdot\alpha\cdot S^{-1},\quad
S\in C^{\infty}(\cE^{\infty})\otimes G \label{gaugetrans}.
\end{equation}
The zero\/-\/curvature representation $\alpha^{S}$ is called
\textsl{gauge\/-\/equivalent} to $\alpha$
and $S$ is the \textsl{gauge transformation}.
Suppose $\alpha = A_i\,dx^i$. The gauge transformation $S$ acts on the
components $A_i$ of $\alpha$ as follows
\begin{equation}
A_i^{S} = \TD{x^i}(S)S^{-1} + SA_iS^{-1}.
   \tag{\ref{gaugetrans}$'$}\label{gaugetransmat}
\end{equation}


\begin{example}[The relation between the coverings which stem from gauge
  equivalent zero curvature representations]
Let $\fg=\mathfrak{sl}_2(\BBC)$ and $G=SL_2(\BBC)$.
Suppose $S\in SL_2(\BBC)$, so that
\[
S = \begin{pmatrix}
  s_1 & s_2 \\
  s_3 & s_4
\end{pmatrix}, \quad \det S = 1.
\]
Let $\alpha = \sum_iA_i dx^i$ be a zero-curvature representation for
an equation $\cE$.
Using decomposition~\eqref{atensmat} for $A_i\in\gsl_2(\BBC)$, we
inspect how the gauge transformation $S$ acts on the components of
$\alpha$:
\begin{multline*}
A_i^S  = \TD{x^i}(S)S^{-1}  + S(a_e^{(i)}\otimes\rho(e) + a_h^{(i)}\otimes\rho(h) a_e^{(i)}\otimes\rho(f) )S^{-1} = \\
{}= \TD{x^i}(S)S^{-1}  + a_e^{(i)}\otimes(S\cdot\rho(e)\cdot S^{-1}) +
a_h^{(i)}\otimes(S\cdot\rho(h)\cdot S^{-1})
+a_e^{(i)}\otimes(S\cdot\rho(f)\cdot S^{-1}),
\end{multline*}
\begin{align*}
  \TD{x^i}(S)S^{-1}  &=
  \begin{pmatrix}
    s_{1;i}s_4 - s_{2;i}s_3 & s_{2;i}s_1 - s_{1;i}s_2\\
    s_{3;i}s_4 - s_{4;i}s_3 & s_{4;i}s_1 - s_{3;i}s_2\\
  \end{pmatrix} =
  \begin{pmatrix}
    s_{1;i}s_4 - s_{2;i}s_3 & s_{2;i}s_1 - s_{1;i}s_2\\
    s_{3;i}s_4 - s_{4;i}s_3 & -s_{1;i}s_4 + s_{2;i}s_3\\
  \end{pmatrix} = \\
{}&{}= (s_{2;i}s_1 - s_{1;i}s_2)\rho(e) + (s_{1;i}s_4 - s_{2;i}s_3)\rho(h) + (s_{3;i}s_4 - s_{4;i}s_3)\rho(f)\\
  S\cdot\rho(e)\cdot S^{-1} & =
  \begin{pmatrix}
    -s_1s_3 & s_1^2 \\
    -s_3^2 & s_1s_3
  \end{pmatrix}  = (s_1^2)\rho(e) + (-s_1s_3)\rho(h) + (-s_3^2)\rho(f),\\
  S\cdot\rho(h)\cdot S^{-1} & =
  \begin{pmatrix}
    s_1s_4+s_2s_3 &  -2s_1s_2 \\
    2s_3s_4 & -s_1s_4 - s_2s_3
  \end{pmatrix} =  (-2s_1s_2)\rho(e) + (s_1s_4+s_2s_3)\rho(h) + (2s_3s_4)\rho(f),\\
  S\cdot\rho(f)\cdot S^{-1} & =
  \begin{pmatrix}
    s_2s_4 & -s_2^2 \\
    s_4^2 & - s_2s_4
  \end{pmatrix} = (-s_2^2)\rho(e) + (s_2s_4)\rho(h) + (s_4^2)\rho(f),
\end{align*}
We finally obtain
\begin{multline*}
  A_i^S = (s_{2;i}s_1 - s_{1;i}s_2 + s_1^2a_e^{(i)}   -2s_1s_2 a_h^{(i)} -s_2^2 a_f^{(i)})\otimes\rho(e)  + {}\\
  {} + (s_{1;i}s_4 - s_{2;i}s_3 - s_1s_3a_e^{(i)} +(s_1s_4+s_2s_3)a_h^{(i)} + s_2s_4a_f^{(i)} )\otimes\rho(h) + {} \\
  {} + (s_{3;i}s_4 - s_{4;i}s_3 -s_3^2a_e^{(i)}  + 2s_3s_4a_h^{(i)} + s_4^2a_f^{(i)})\otimes\rho(f).
\end{multline*}
Passing to the vector field representation of $A^{S}_i$ by using
formula~\eqref{atensvf}, we have
\begin{multline}
\label{asvf}
 V_{A_i^S} = (s_{2;i}s_1 - s_{1;i}s_2 + s_1^2a_e^{(i)}   -2s_1s_2 a_h^{(i)} -s_2^2 a_f^{(i)})\otimes\varrho(e)  + {}\\
  {} + (s_{1;i}s_4 - s_{2;i}s_3 - s_1s_3a_e^{(i)} +(s_1s_4+s_2s_3)a_h^{(i)} + s_2s_4a_f^{(i)} )\otimes\varrho(h) + {} \\
  {} + (s_{3;i}s_4 - s_{4;i}s_3 -s_3^2a_e^{(i)}  + 2s_3s_4a_h^{(i)} + s_4^2a_f^{(i)})\otimes\varrho(f).
\end{multline}
In other words, whenever we start from the covering of $\cE$
associated with a zero\/-\/curvature representation $\alpha$, such that
the differentiation rules for the nonlocality $w$ are
\[
\TD{x^i}(w) = - a_{e}^{(i)} + 2a_{h}^{(i)}w +a_f^{(i)}w^2,
\]
we obtain the 
covering which is associated with $\alpha^{S}$:
\begin{multline}
\label{aicov}
\TD{x^i}(w_S) = - (s_{2;i}s_1 - s_{1;i}s_2 + s_1^2a_e^{(i)}   -2s_1s_2 a_h^{(i)} -s_2^2 a_f^{(i)})  + {}\\
  {} + 2(s_{1;i}s_4 - s_{2;i}s_3 - s_1s_3a_e^{(i)} +(s_1s_4+s_2s_3)a_h^{(i)} + s_2s_4a_f^{(i)} )w_S + {} \\
  {} + (s_{3;i}s_4 - s_{4;i}s_3 -s_3^2a_e^{(i)}  + 2s_3s_4a_h^{(i)} + s_4^2a_f^{(i)})w^2_S.
\end{multline}

We shall use this relations between the two coverings
in the search of the gauge transformations between known
zero\/-\/curvature representations for the KdV equation.
\end{example}

\begin{example}[Gauge transformations between zero\/-\/curvature
  representations for the KdV equation]\label{exBVV2Gardner}
Let us find the gauge transformations that bring
coverings~\eqref{mcov} and~\eqref{miuracov} to the form~\eqref{covBVV}.

For the transformation~\eqref{mcov}$\to$\eqref{covBVV} we have
\begin{align*}
p_x  & = u + p^2 - \lambda = -(s_{2;x}s_1 - s_{1;x}s_2 + s_1^2   -s_2^2 (\lambda-u)  + {}\\
{} & {} - 2(s_{1;x}s_4 - s_{2;x}s_3 - s_1s_3  + s_2s_4(\lambda-u) )p - {} \\
{} & {} - (s_{3;x}s_4 - s_{4;x}s_3 -s_3^2  + s_4^2(\lambda-u))p^2).
\end{align*}
Solving this equation for $s_i$, we find a unique solution $s_2 = s_3
= i$, $s_1=s_4=0$:
\begin{equation}
\label{gaugeM2BVV}
S=\begin{pmatrix}
0 & i \\
i & 0
\end{pmatrix},
\quad S^{-1}=\begin{pmatrix}
0 & -i \\
-i & 0
\end{pmatrix}.
\end{equation}
The matrices of the zero curvature representations~\eqref{mcov}
and~\eqref{covBVV} are related as follows:
\[
\begin{pmatrix}
0 & i \\
i & 0
\end{pmatrix}
\begin{pmatrix}
0 & 1 \\
\lambda - u & 0
\end{pmatrix}
\begin{pmatrix}
0 & -i \\
-i & 0
\end{pmatrix}
=
\begin{pmatrix}
0 & \lambda - u\\
1 & 0
\end{pmatrix}.
\]

On the other hand, for the
transformation~\eqref{miuracov}$\to$\eqref{covBVV} we have
\begin{align*}
  p_x  & = u + p^2 - \frac{1}{4\veps^2} = - (s_{2;x}s_1 - s_{1;x}s_2 - s_1^2\frac{u}{\veps}   + s_1s_2\frac{1}{\veps} -s_2^2 \veps  + {}\\
  {}&  {} - 2(s_{1;x}s_4 - s_{2;x}s_3 + s_1s_3\frac{u}{\veps} -(s_1s_4+s_2s_3)\frac{1}{2\veps} + s_2s_4\veps )p - {} \\
  {}& {} - (s_{3;x}s_4 - s_{4;x}s_3 + s_3^2\frac{u}{\veps} -
  s_3s_4\frac{1}{\veps} + s_4^2\veps)p^2).
\end{align*}
Solving this equation for $s_i$, we find a solution
\(s_1 = i/\sqrt{\veps}\), \(s_2 = i/(2\veps\sqrt{\veps})\),
 \(s_3 =0\),  \(s_4 = i\sqrt{\veps}\). Therefore,
\begin{equation}
\label{gaugeMiura2BVV}
S=\begin{pmatrix}
  i/\sqrt{\veps} & i/(2\veps\sqrt{\veps})\\
  0 & -i\sqrt{\veps}
\end{pmatrix},
\quad
S^{-1} = \begin{pmatrix}
  -i\sqrt{\veps} & -i/(2\veps\sqrt{\veps})\\
  0 & i/\sqrt{\veps}
\end{pmatrix},
\end{equation}
The matrices of the zero\/-\/curvature representations~\eqref{miuracov}
and~\eqref{covBVV} satisfy the relation
\[
\begin{pmatrix}
  i/\sqrt{\veps} & i/(2\veps\sqrt{\veps})\\
  0 & -i\sqrt{\veps}
\end{pmatrix}
\begin{pmatrix}
    0 & \frac{1}{4\veps^2} - u \\
    1 & 0
\end{pmatrix}
\begin{pmatrix}
  -i\sqrt{\veps} & -i/(2\veps\sqrt{\veps})\\
  0 & i/\sqrt{\veps}
\end{pmatrix}
=
\begin{pmatrix}
\frac{1}{2\veps} & \frac{u}{\veps} \\
-\veps & -\frac{1}{2\veps}
\end{pmatrix}.
\]
\end{example}

Let us recall that in Example~\ref{exLaxZCR} we derived the
zero\/-\/curvature representation for the KdV equation from its Lax
pair. Having done that, we also revised the transition from this
zero\/-\/curvature representation to the Gardner deformation of the KdV
equation. In the next section we extend this approach and find the
generalizations of Gardner's deformation~\eqref{kdvdef} for
graded systems,
in particular, for the $N{=}1$ and $N{=}2$ supersymmetric 
Korteweg\/--\/de Vries equations.

\section{Graded systems}\label{SecGraded}
\subsection{Lie super\/-\/algebras}
We recall first the definition of the Lie
super\/-\/algebra~\cite{Lejtes84,Manin84,Berezin}.
Let $\mathcal{A}$ be an algebra over the field \(\BBC\) and
$\BBZ_2=\BBZ/2\BBZ=\{\bar{0},\bar{1}\}$ be the group of residues
modulo~$2$.
An algebra $\mathcal{A}$ is called a super\/-\/algebra if $\mathcal{A}$
can be decomposed as the direct sum \(\mathcal{A} =
\mathcal{A}_{\bar{0}}\oplus\mathcal{A}_{\bar{1}}\) such that
\[
\mathcal{A}_{\bar{0}}\cdot\mathcal{A}_{\bar{0}}\subset\mathcal{A}_{\bar{0}},
\quad
\mathcal{A}_{\bar{0}}\cdot\mathcal{A}_{\bar{1}}\subset\mathcal{A}_{\bar{1}},
\quad
\mathcal{A}_{\bar{1}}\cdot\mathcal{A}_{\bar{1}}\subset\mathcal{A}_{\bar{0}}.
\]
A nonzero element of $\mathcal{A}_{\bar{0}}$ or
$\mathcal{A}_{\bar{1}}$ is called \textsl{homogeneous} (respectively,
even or odd). Let $p(a)=k$ if $a\in\mathcal{A}_{k}$ for $k\in\BBZ_2$.
The number $p(a)$ is the parity of~$a$.

The super\/-\/algebra $\fg$ is a \textsl{Lie super\/-\/algebra}
if it is endowed with the multiplication $\bbl \cdot\,,\cdot \bbr$
that satisfies the equalities
\begin{gather}
\bbl x, y \bbr = - (-1)^{p(x)p(y)} \bbl y,x \bbr, \label{scomm1}\\
\bbl x, \bbl y, z \bbr\bbr = \bbl\bbl x, y \bbr, z \bbr +
(-1)^{p(x)p(y)} \bbl y, \bbl x,z \bbr \bbr.
\end{gather}
here $x$,\ $y$, and $z$ are arbitrary elements of $\mathcal{A}$ and
$x$,\ $y$ are presumed homogeneous.

The super\/-\/matrix structure of a matrix is achieved whenever the
parity is assigned to its rows and columns.
We choose the super\/-\/matrix structure such that rows
(respectively, columns) which are assigned the even parity always
precede the rows (columns) of odd parity~\cite{Lejtes84}.
If a matrix has $r$ even and $s$ odd rows and $p$ even and $q$
odd columns, then its dimension is said to be
equal to $(r\mid s)\times(p \mid q)$.
In particular, we shall use the shorthand notation $(p\mid q)$ for
the dimension $(p\mid q)\times(p\mid q)$.
We denote by $\Mat(p\mid q;\mathcal{A})$ the set of all matrices of
dimension $(p\mid q)$ with elements that belong the
super\/-\/algebra~$\mathcal{A}$.

Let us introduce the super\/-\/matrix structure on the space
$\Mat(p\mid q;\mathcal{A})$. Consider a matrix
$X=\left(\begin{smallmatrix}
  R & S \\
  T & U
\end{smallmatrix}\right)\in\Mat(p\mid q;\mathcal{A})$ and set
\begin{align*}
p(X) = \bar{0} & \quad \text{if }
  p(R_{ij})=p(U_{ij})=\bar{0},\quad p(T_{ij})=p(S_{ij})=\bar{1};\\
p(X) = \bar{1} & \quad \text{if }
  p(R_{ij})=p(U_{ij})=\bar{1},\quad p(T_{ij})=p(S_{ij})=\bar{0}.
\end{align*}
Taking into account the graded skew\/-\/symmetry~\eqref{scomm1} of the
bracket $\bbl\cdot,\cdot\bbr$, we define the Lie
super\/-\/algebra structure on the space $\Mat(p\mid q; \mathcal{A})$ by
the formula
\begin{equation}\label{scomm}
\bbl X, Y \bbr = XY - (-1)^{p(X)p(Y)}YX,\quad X,Y\in\Mat(p \mid q;\mathcal{A}).
\end{equation}
The Lie super\/-\/algebras $\gl(m \mid n)=\Mat(m \mid n,\BBC)$ and
$\gsl(m \mid n)=\{X\in\gl(m \mid n)|\str X=0\}$, where
$\str\left(\begin{smallmatrix}
  R & S \\
  T & U
\end{smallmatrix}\right) = \tr R - \tr U$,
are called the general linear and special linear Lie super\/-\/algebras,
respectively.

To calculate the super\/-\/commutator $\bbl X,Y \bbr$ of two
nonhomogeneous elements $X$ and $Y$, we first split $X = X_{\bzero} +
X_{\bone}$ and $Y = Y_{\bzero} + Y_{\bone}$ so that
$p(X_{\bzero})=p(Y_{\bzero})=\bar{0}$ and
$p(X_{\bone})=p(Y_{\bone})=\bar{1}$. Using~\eqref{scomm}, we obtain
\begin{multline}\label{scommexpad}
\bbl X, Y\bbr  = \bbl X_{\bzero} + X_{\bone}, Y_{\bzero} + Y_{\bone}\bbr  =
\bbl X_{\bzero}, X_{\bzero}\bbr  + \bbl X_{\bzero}, Y_{\bone}\bbr  +
\bbl X_{\bone}, Y_{\bzero}\bbr  + \bbl X_{\bone}, Y_{\bone}\bbr  = {}\\
{} = (X_{\bzero}Y_{\bzero} - Y_{\bzero}X_{\bzero}) + (X_{\bzero}Y_{\bone} - Y_{\bone}X_{\bzero}) + (X_{\bone}Y_{\bzero} - Y_{\bzero}X_{\bone}) +
(X_{\bone}Y_{\bone} + Y_{\bone}X_{\bone}).
\end{multline}

The super\/-\/determinant, or the \textsl{Berezinian} of an invertible
matrix
$X=\left(\begin{smallmatrix}
  R & S \\
  T & U
\end{smallmatrix}\right)\in \gl(m\mid n)$
is given by the formula~\cite{Berezin}
\[
\sdet \begin{pmatrix}
  R & S \\
  T & U
\end{pmatrix} = \frac{\det(R - SU^{-1}T)}{\det{U}}.
\]

\begin{example}
In what follows, we shall use the Lie super-algebra
$\gsl(1\mid2)\simeq\gsl(2\mid1)$, see~\cite{Frappat:1996pb}.
Its representation in the space $\Mat(2\mid1;\BBC)$ is given by
the eight basic vectors,
four even: $E^{+}$, $E^{-}$, $H$, and $Z$, and four odd:
$F^{+}$, $F^{-}$, $\bF^{+}$, and $\bF^{-}$, where
\[
\begin{array}{llll}
  E^{+} =
  \begin{pmatrix}
    0 & 1 & 0 \\
    0 & 0 & 0 \\
    0 & 0 & 0
  \end{pmatrix} &
  E^{-} =
  \begin{pmatrix}
    0 & 0 & 0 \\
    1 & 0 & 0 \\
    0 & 0 & 0
  \end{pmatrix} &
  H =
  \begin{pmatrix}
    1/2 & 0 & 0 \\
    0 & -1/2 & 0 \\
    0 & 0 & 0
  \end{pmatrix} &
  Z =
  \begin{pmatrix}
    1/2 & 0 & 0 \\
    0 & 1/2 & 0 \\
    0 & 0 & 1
  \end{pmatrix} \\
  F^{+} =
  \begin{pmatrix}
    0 & 0 & 0 \\
    0 & 0 & 0 \\
    0 & 1 & 0
  \end{pmatrix} &
  F^{-} =
  \begin{pmatrix}
    0 & 0 & 0 \\
    0 & 0 & 0 \\
    1 & 0 & 0
  \end{pmatrix} &
  \bar{F}^{+} =
  \begin{pmatrix}
    0 & 0 & 1 \\
    0 & 0 & 0 \\
    0 & 0 & 0
  \end{pmatrix} &
  \bar{F}^{-} =
  \begin{pmatrix}
    0 & 0 & 0 \\
    0 & 0 & 1 \\
    0 & 0 & 0
  \end{pmatrix}.
\end{array}
\]
The elements of the basis satisfy the following commutation
relations: 
\begin{align*}
\bbl H,E^\pm \bbr &= \pm E^\pm &
  \bbl H,F^\pm \bbr &= \pm \half F^\pm &
     \bbl H,\bF^\pm \bbr &= \pm \half \bF^\pm \\
\bbl Z,H \bbr &= \bbl Z,E^\pm \bbr = 0 &
  \bbl Z,F^\pm \bbr  &= \half F^\pm &
     \bbl Z,\bF^\pm \bbr &= -\half \bF^\pm \\
\bbl E^\pm,F^\pm \bbr &= \bbl E^\pm,\bF^\pm \bbr = 0 &
   \bbl E^\pm,F^\mp \bbr &= -F^\pm &
     \bbl E^\pm,\bF^\mp \bbr  &= \bF^\pm \\
\bbl F^\pm,F^\pm \bbr &= \bbl \bF^\pm,\bF^\pm \bbr = 0 &
  \bbl F^\pm,F^\mp \bbr &= \bbl \bF^\pm,\bF^\mp \bbr = 0 &
    \bbl F^\pm,\bF^\pm \bbr &= E^\pm \\
\bbl E^+,E^- \bbr &= 2H &
   \bbl F^\pm,\bF^\mp \bbr &= Z \mp H. &
\end{align*}
The Lie super\/-\/algebra \(\gsl(2\mid1)\) contains the Lie algebra
\(\gsl(2,\BBC)\) as a subalgebra. The vectors \(E^{\pm}\) and \(H\)
form the basis in~\(\gsl(2,\BBC)\).

The Lie super\/-\/group $SL(2\mid 1)$, which corresponds to the Lie
super\/-\/algebra $\gsl(2\mid 1)$, consist of the matrices with unit
Berezinian: $SL(2\mid1) = \{ S\in GL(2\mid 1) \mid \sdet S= 1\}$.
\end{example}

\begin{rem}
\label{remSL21dec}
Consider the following three subgroups 
of the Lie super-group $SL(2\mid 1)$:
\[
G_{+}=\left\{\begin{pmatrix} \mathbf{1} & B \\  0 & 1 \end{pmatrix}\right\},\quad
G_{0}=\left\{\begin{pmatrix} A & 0 \\  0 & D \end{pmatrix}\right\},\quad
G_{-}=\left\{\begin{pmatrix} \mathbf{1} & 0 \\  C & 1 \end{pmatrix}\right\}.
\]
Each matrix \(S\in SL(2\mid1)\) can be represented~\cite{Lejtes80} as a
product $S=S_{+}S_{0}S_{-}$, where $S_{+}\in G_{+}$, $S_{0}\in
G_{0}$, $S_{-}\in G_{-}$. Due to the multiplicativity of the
Berezinian, $\sdet S = \sdet S_{+}\cdot\sdet S_{0}\cdot\sdet S_{-} =
1$, and
in view of the obvious property $\sdet S_{+}=\sdet S_{-} = 1$ for all
elements of the groups $G_{+}$ and $G_{-}$,
we conclude that $\sdet S_0 = 1$ for all $S_0 \in G_0$.

For the Lie super\/-\/group $SL(2\mid1)$, the dimension of
the matrix $D$ is equal to $1\times 1$ and the dimension of the matrix
$A$ is equal to $2\times 2$. Let us show that $G_{0} \simeq
SL(2\mid0)$.
The condition $\sdet S_{0} = 1$ for the matrix  $S_0\in SL(2\mid1)$
implies the equality $\det A = \det D$ of the usual
determinants of $A$ and $D$. Therefore, to each matrix $A\in
GL(2\mid0)$ we can put into correspondence the matrix
$S_A \in G_{0}$ by setting $S_A
= \left(\begin{smallmatrix} A & \boldsymbol{0} \\ \boldsymbol{0} & \det
A \end{smallmatrix}\right)$ and conversely, to each matrix $S =
\left(\begin{smallmatrix} A & \boldsymbol{0} \\ \boldsymbol{0} & D
\end{smallmatrix}\right)\in SL(2\mid1)$ we can associate the matrix $A$
from~$GL(2\mid0)$.
\end{rem}

\subsection{Zero\/-\/curvature representations of graded extension of
the KdV equation}
The graded extension of the Maurer\/--\/Cartan
equation~\ref{zcrm} has the form
\begin{equation}\label{zcrmsuper}
\TD{x^j} A_i - \TD{x^i} A_j + \bbl A_i, A_j\bbr  = 0, \quad
   \forall i,j=1,\dots,m: i\neq j.
\end{equation}
Let us study in more detail the geometry of the $N{=}1$ and $N{=}2$
supersymmetry\/-\/invariant generalizations of the
Korteweg\/--\/de Vries equation~\cite{MathieuNew,MathieuN1}.

\subsubsection{$N=1$ supersymmetric Korteweg\/--\/de Vries equation}
The $N=1$ supersymmetric generalization of the KdV
equation~\eqref{kdv} is the sKdV equation~\cite{MathieuN1}
\begin{equation}\label{sKdV}
\phi_t = - \phi_{xxx} - 3(\phi \sD \phi)_x, \quad \sD =
\frac{\partial}{\partial \theta} + \theta\frac{d}{dx},
\end{equation}
where $\phi(x,t,\theta) = \xi + \theta u$ is a complex fermionic
super-field,  \(\theta\) is the Grassmann (or anti\/-\/commuting)
variable such that \(\theta^2=0\), the unknown \(u\) is the bosonic
field, and \(\xi\) is the fermionic field.
By using the expansion $\phi(x,t,\theta) = \xi + \theta u$
in~\eqref{sKdV}, we obtain
\begin{subequations}\label{sKdVComponents}
\begin{align}
\underline{u_t} = & \underline{- u_{xxx} - 6uu_x} + 3\xi\xi_{xx}, \label{n1getkdv}\\
\xi_t = & - \xi_{xxx} - 3(u\xi)_x.
\end{align}
\end{subequations}
The KdV equation~\eqref{kdv} is underlined in~\eqref{n1getkdv}.

\begin{example}[Zero\/-\/curvature representation and Gardner's
  deformation of the sKdV equation]\label{n1example}
The sKdV equation~\eqref{sKdVComponents} admits the $\gsl(2\mid
1)$-valued zero\/-\/curvature representation
\[
\alpha^{N=1} = A^{N=1}_1 dx + B^{N=1}_1 dt,
\]
where
\[
A^{N=1}_1 = \begin{pmatrix}
-\frac{1}{2\veps} & - u +\frac{1}{4\veps^2} & \xi \\
1 & -\frac{1}{2\veps} & 0 \\
0 & -\xi & -\frac{1}{\veps}
\end{pmatrix},
\]
\[
B^{N=1}_1 = \begin{pmatrix}
 \frac12\veps^{-3} - u_{x}  &
   2u^2 + u_{xx} - \xi \xi_{x} + \frac12\veps^{-2}u  -
   \frac14\veps^{-4}&
     - \xi_{xx} - 2\xi u  - \frac12\veps^{-1} \xi_{x} - \frac12\veps^{-2}\xi \\
 - 2u  - \veps^{-2} &
   \frac12\veps^{-3} + u_{x} &
     - \xi_{x} - \xi \veps^{-1} \\
 - \xi_{x} + \xi \veps^{-1}  &
   \xi_{xx} + 2\xi u  - \frac12\veps^{-1} \xi_{x} +  \frac12\veps^{-2}\xi &
     \veps^{-3}
\end{pmatrix}.
\]

Let us construct the generalization $S^{N=1}\in SL(2\mid1)$ of
gauge transformation~\eqref{gaugeMiura2BVV} where we had $S\in
SL_2(\BBC)\simeq SL(2\mid0)$. Taking into account
Remark~\ref{remSL21dec}, we consider the ansatz
$S^{N=1}=S^{N=1}_{+}S^{N=1}_{0}S^{N=1}_{-}$,
where $S_{\alpha}\in G_{\alpha}$, $\alpha\in\{+,0,-\}$.
Bearing in mind that $SL_2(\BBC)\simeq
SL(2\mid0)\subset{GL(2\mid0)}$, we construct $S$ by using following
scheme:
\begin{enumerate}
\item we obtain an element $S^{N=1}_{0}$ by
the multiplication of $S$ from
right and left by some matrices from $GL(2|0)$;
\item we specify the matrices $S^{N=1}_{+}$ and $S^{N=1}_{-}$.
\end{enumerate}
We construct the matrix $S^{N=1}$ as follows
\begin{multline}
\label{n1gauge}
S^{N=1} =  \begin{pmatrix}
-1 & -\frac12\veps^{-1} & 0 \\
0 & \veps & 0 \\
0 & 0 & -\veps
\end{pmatrix} =
\\
{}=\underbrace{\begin{pmatrix}
  1 & 0 & 0 \\
  0 & 1 & 0 \\
  0 & 0 & 1 \\
\end{pmatrix}}_{S^{N=1}_{+}}
\underbrace{\begin{pmatrix}
i\sqrt{\veps} & i\sqrt{\veps}/\veps^2 &  0 \\
0 & i\sqrt{\veps} & 0 \\
0 & 0 & - \veps
\end{pmatrix}
\overbrace{\begin{pmatrix}
  i/\sqrt{\veps} & i/(2\veps\sqrt{\veps}) & 0\\
  0 & -i\sqrt{\veps} & 0 \\
  1 & 0 & 1 \\
\end{pmatrix}}^{S}
\begin{pmatrix}
1 & \veps^{-1} &  0 \\
0 & 1 & 0 \\
0 & 0 & 1
\end{pmatrix}
}_{S^{N=1}_{0}}
\underbrace{\begin{pmatrix}
  1 & 0 & 0 \\
  0 & 1 & 0 \\
  0 & 0 & 1 \\
\end{pmatrix}}_{S^{N=1}_{-}}.
\end{multline}

By applying the gauge transformation $S^{N=1}$ to the zero\/-\/curvature
representation $\alpha^{N=1}$, we obtain the gauge\/-\/equivalent
zero\/-\/curvature representation $\beta$ for the sKdV
equation~\eqref{sKdVComponents}:
\begin{equation}\label{n1beta}
\beta^{N=1} = (\alpha^{N=1})^{S^{N=1}} = A^{N=1}_2 dx + B^{N=1}_2 dt,
\end{equation}
where
\[
A^{N=1}_2 = \begin{pmatrix}
  0 & \veps^{-1} u  & \veps^{-1}\xi \\
  - \veps &  - \veps^{-1} &  0 \\
  0 & \xi & - \veps^{-1}
\end{pmatrix},
\]
\[
B^{N=1}_2 = \begin{pmatrix}
  u_{x} - u\veps^{-1} &
    \frac{1}{\veps}( - 2u^2 - u_{xx} + \xi \xi_{x})
- \frac{1}{\veps^{2}} u_{x} - \frac{1}{\veps^{3}}u &
    \frac{1}{\veps}( - \xi_{xx} - 2\xi u )
- \frac{1}{\veps^{2}} \xi_{x} - \frac{1}{\veps^{3}}\xi\\
2u \veps + \veps^{-1} &
   u_{x} + u \veps^{-1} + \veps^{-3} &
     \xi_{x} + \xi \veps^{-1} \\
- \xi_{;x}\veps  + \xi &
   - \xi_{xx} - 2\xi u &
     \veps^{-3}
\end{pmatrix}.
\]

Let us recall that formula~\eqref{wroelofs} yields the representation
of the matrices $A^{N=1}_2$ and $B^{N=1}_2$ in terms of vector
fields. By this argument, from the zero\/-\/curvature representation
$\beta^{N=1}$ we obtain the two\/-\/dimensional covering over the sKdV
equation~\eqref{sKdVComponents}; one of the two new nonlocal variables
is bosonic (let us denote it by~$\tu$)
and the other, $\txi$ is fermionic:
\begin{align*}
\tu_{x}  = {}& - \tu^2\veps + (\tu - u)\veps^{-1}  - \txi\xi,\\
\txi_{x}  = {} & - \txi\tu\veps + (\txi - \xi )\veps^{-1},\\
\tu_{t}  = {} &\frac{1}{\veps^3}(2\tu^2u \veps^4 + \tu^2\veps^2 - 2\tu u
\veps^2 - 2\tu u_{x}\veps^3 - \tu + 2u^2\veps^2 + u  + u_{xx}\veps^2 +
   u_{x}\veps  - \txi\tu\xi_{x}\veps^4 + \txi\xi_{xx}\veps^3 + {}\\
{}&{}+ \txi\xi \tu\veps^3 + 2\txi\xi u \veps^3 - \xi
\xi_{x}\veps^2),\\
\txi_t = {}& \frac{1}{\veps^3}( - \tu\xi_{x}\veps^3 + \xi_{xx}\veps^2 +
  \xi_{x}\veps  + 2\txi\tu u \veps^4 + \txi\tu\veps^2 - \txi u \veps^2
  - \txi u_{x}\veps^3 - \txi - \xi \tu\veps^2 + 2\xi u \veps^2 + \xi ).
\end{align*}
We now
express the local variables $u$ and $\xi$ from $\tu_x$ and $\txi_x$
and substitute them in $\tu_t$ and $\txi_t$. We thus obtain the
Gardner deformation~\cite{MathieuOpen} of the sKdV
equation~\eqref{sKdVComponents}:
\begin{align*}
\cE_{\veps} = \, \Bigr\{ & \tu_{t}=  \underline{6\tu^2\tu_x\veps^2 - 6\tu\tu_x - \tu_{xxx}} -
  3\txi\tu\txi_{xx}\veps^2 + 3\txi\txi{xx}- 3\txi\txi_x\tu_x\veps^2,\\
{}&\txi_t=  3\tu^2\txi_x\veps^2 - 3\tu\txi_x - \txi_{xxx} +
3\txi\tu\tu_x\veps^2 - 3\txi\tu_x \Bigl\},\\
\mathfrak{m}_{\veps} = \, \Bigr\{ & u =   \underline{\tu - \veps\tu_x}  + \veps^2(\txi\txi_x\veps^2 \underline{- \tu^2)},
\quad
\xi = \txi - \veps\txi_x  - \veps^2\txi\tu \Bigl\}:
\cE_{\veps}\to\cE_{\text{sKdV}}.
\end{align*}
The original Gardner deformation~\eqref{kdvdef} of the KdV
equation~\eqref{kdv} is underlined in the above formulas.

In Example~\ref{n1example} we generalized, to the graded case,
the scheme
of constructing Gardner's deformations from known zero\/-\/curvature
representations (see Example~\ref{exBVV2Gardner}).
\end{example}

\subsubsection{$N=2$ supersymmetric Korteweg\/--\/de Vries equation}
Let us consider the four\/-\/component generalization of the KdV
equation~\eqref{kdv}, namely, the $N{=}2$
supersymmetric Korteweg\/--\/de Vries equation (SKdV)~\cite{MathieuNew}:
\begin{equation}\label{SKdV}
\bu_t=-\bu_{xxx}+3\bigl(\bu\cD_1\cD_2\bu\bigr)_x
 +\frac{a-1}{2}\bigl(\cD_1\cD_2\bu^2\bigr)_x + 3a\bu^2\bu_x,\qquad
 \cD_i=\frac{\dd}{\dd\theta_i}+\theta_i\cdot\frac{\Id}{\Id x},
\end{equation}
where
\begin{equation}\label{N2superfield}
\bu(x,t;\theta_1,\theta_2)
=u_0(x,t)+\theta_1\cdot u_1(x,t)+\theta_2\cdot
u_2(x,t)+\theta_1\theta_2 \cdot u_{12}(x,t)
\end{equation}
is the complex bosonic super\/-\/field,
$\theta_1,\theta_2$ are Grassmann variables such that $\theta_1^2=\theta_2^2=\theta_1\theta_2 + \theta_2\theta_1 =
0$, $u_0$, $u_{12}$ are bosonic fields, and  $u_1$, $u_2$ are
fermionic fields. Expansion~\eqref{N2superfield}
converts~\eqref{SKdV} to the four\/-\/component system
\begin{subequations}\label{SKdVComponents}
\begin{align}
u_{0;t}&=-u_{0;xxx}+\bigl(a u_0^3
   -(a+2)u_0u_{12}+(a-1)u_1u_2\bigr)_x,
\label{GetmKdV}\\
u_{1;t}&=-u_{1;xxx}+\bigl(\phantom{+}(a+2)u_0u_{2;x}+(a-1)u_{0;x}u_2
   -3u_1u_{12}+3a u_0^2u_1 \bigr)_x,\\
u_{2;t}&=-u_{2;xxx}+\bigl(-(a+2)u_0u_{1;x}-(a-1)u_{0;x}u_1
   -3u_2u_{12}+3a u_0^2u_2 \bigr)_x,\\
\underline{u_{12;t}}&=\underline{-u_{12;xxx}-6u_{12}u_{12;x}}
 +3au_{0;x}u_{0;xx}+(a+2)u_0u_{0;xxx}\notag\\
 {}&{}\qquad{}+3u_1u_{1;xx}+3u_2u_{2;xx}
 +3a\bigl(u_0^2u_{12} -2u_0u_1u_2\bigr)_x.\label{GetKdV}
\end{align}
\end{subequations}
The KdV equation is underlined in~\eqref{GetKdV}. The SKdV equation is
most interesting (in particular, bi\/-\/Hamiltonian, whence
completely integrable) if $a\in\{-2, 1,4\}$,
see~\cite{MathieuNew,HKKW,KisHus09}.
Let us consider the bosonic limit $u_1=u_2=0$ of
system~\eqref{SKdVComponents}: 
by setting
$a=-2$ we obtain the triangular system which consists of the modified
KdV equation upon $u_0$ and the equation of KdV\/-\/type; in the
case $a=1$ we obtain the Krasil'shchik\/--\/Kersten system; for $a=4$,
we obtain the third equation in the Kaup\/--\/Boussinesq
hierarchy. A Gardner deformation of the Kaup\/--\/Boussinesq system
was constructed in~\cite{HKKW}.

The Gardner deformation problem for the $N=2$ supersymmetric $a=4$
KdV equation was formulated by P.~Mathieu in~\cite{MathieuNew}.
In the paper~\cite{HKKW} it was shown that one can not construct
such a deformation under the assumptions that, first, the
deformation is polynomial in $\cE$, second, it involves only the
super\/-\/fields but not their components, and third, it contains the
known deformation~\eqref{kdvdef} under the reduction $u_0=0$, $u_1=u_2=0$.
Therefore, we shall find a graded generalization of Gardner's
deformation~\eqref{kdvdef} for the system of four
equations~\eqref{SKdVComponents} treating it in components but not as
the single equation~\eqref{SKdV} upon the super\/-\/field.

The SKdV equation~\eqref{SKdVComponents} admits~\cite{Das} the
$\gsl(1\mid 2)$-\/valued zero\/-\/curvature representation
$\alpha^{N=2}= A dx + B dt$ such that
\begin{subequations}
\label{daszcr}
\begin{gather}
A = \begin{pmatrix}
    \eta  - iu_{0}  &
  \eta^2 - 2i\eta u_{0} - u_{0}^2 - u_{12}  &
  - u_{2} - iu_{1} \\
  1 &
  \eta  - iu_{0} &
  0 \\
  0 &
  u_{2} - u_{1}i &
  2\eta  - 2iu_{0}
\end{pmatrix},\\
B = \begin{pmatrix}
b_{11} & b_{12} & b_{13}  \\
b_{21} & b_{22} & b_{23}  \\
b_{31} & b_{32} & b_{33}
\end{pmatrix},
\end{gather}
\end{subequations}
where the elements of $B$ are as follows:
\begin{align*}
b_{11} = & - 4\eta^3 - 2\eta u_{0;x}i - 4u_{0}^3i + 6u_{0} u_{12} i +
4u_{0} u_{0;x} + u_{0;xx}i - u_{12;x} + 4u_{2}u_{1}i, \\
b_{12} = &  - 4\eta^4 + 4\eta^3u_{0} i + 2\eta^2u_{12}  - 4\eta
u_{0}^3i + 8\eta u_{0} u_{12} i + 2\eta u_{0;xx}i - 4u_{0}^4 -
2u_{0}^2u_{12}  - {} \\
{}&{}- 4u_{0} u_{0;xx} + 2u_{12}^2 - 4u_{0;x}^2 + u_{12;xx} -
u_{2}u_{2;x} + 4u_{2}u_{1}\eta i - 8u_{2}u_{1}u_{0}  - u_{1}u_{1;x},\\
b_{13} = & - \eta u_{2;x} - \eta u_{1;x}i - 5u_{0} u_{2;x}i + 5u_{0}
u_{1;x} + u_{2;xx} + u_{1;xx}i + 2u_{2}\eta^2 + 2u_{2}\eta u_{0} i - {}\\
{}&{}-8u_{2}u_{0}^2 + 2u_{2}u_{12}  - 4u_{2}u_{0;x}i + 2u_{1}\eta^2i -
2u_{1}\eta u_{0}  - 8u_{1}u_{0}^2i + 2u_{1}u_{12} i + 4u_{1}u_{0;x},\\
b_{21} = &  2( - 2\eta^2 - 2\eta u_{0} i + 2u_{0}^2 - u_{12} ) ,\\
b_{22} = &  - 4\eta^3 + 2\eta u_{0;x}i - 4u_{0}^3i + 6u_{0} u_{12} i - 4u_{0} u_{0;x} + u_{0;xx}i + u_{12;x} + 4u_{2}u_{1}i,\\
b_{23} = & u_{2;x} + u_{1;x}i - 2u_{2}\eta  - 4u_{2}u_{0} i -
2u_{1}\eta i + 4u_{1}u_{0},\\
b_{31} = & u_{2;x} - u_{1;x}i + 2u_{2}\eta  + 4u_{2}u_{0} i - 2u_{1}\eta i + 4u_{1}u_{0},\\
b_{32} = &  - \eta u_{2;x} + \eta u_{1;x}i - 5u_{0} u_{2;x}i - 5u_{0}
u_{1;x} - u_{2;xx} + u_{1;xx}i - 2u_{2}\eta^2 - 2u_{2}\eta u_{0} i
+{}\\
{}&{}+ 8u_{2}u_{0}^2 - 2u_{2}u_{12}  - 4u_{2}u_{0;x}i + 2u_{1}\eta^2i - 2u_{1}\eta u_{0}  - 8u_{1}u_{0}^2i + 2u_{1}u_{12} i - 4u_{1}u_{0;x},\\
b_{33} = & 2( - 4\eta^3 - 4u_{0}^3i + 6u_{0} u_{12} i + u_{0;xx}i + 4u_{2}u_{1}i).
\end{align*}

\begin{rem}\label{remz}
Let us recall that the vectors $Z$, $H$ and $E^{\pm}$ that belongs
to $\gsl(2\mid 1)$,
\[Z=\begin{pmatrix}
\frac12 & 0 & 0\\
0 & \frac12 & 0\\
0 & 0 & 1\\
\end{pmatrix},\quad
  E^{+} =
  \begin{pmatrix}
    0 & 1 & 0 \\
    0 & 0 & 0 \\
    0 & 0 & 0
  \end{pmatrix},\quad
  E^{-} =
  \begin{pmatrix}
    0 & 0 & 0 \\
    1 & 0 & 0 \\
    0 & 0 & 0
  \end{pmatrix},\quad
  H =
  \begin{pmatrix}
    \frac12 & 0 & 0 \\
    0 & -\frac12 & 0 \\
    0 & 0 & 0
  \end{pmatrix},
\]
generate a basis in $\gl(2,\BBC)$. The vector $Z$ commutes with
any other vector from $\gl(2,\BBC)$.
\end{rem}

The reduction $u_0=u_1=u_2=0$ converts zero\/-\/curvature
representation~\eqref{daszcr} to the $\gl(2,\BBC)$-\/valued
zero\/-\/curvature representation of the KdV equation~\eqref{kdv},
\begin{align*}
A_{\text{KdV}} = &\begin{pmatrix}
\eta & \eta^2 - u_{12} & 0 \\
1 & \eta & 0 \\
0 & 0 & 2\eta
\end{pmatrix},\\
B_{\text{KdV}} = & \begin{pmatrix}
- 4\eta^3 - u_{12;x} & - 4\eta^4 + 2\eta^2u_{12}  + 2u_{12}^2 +
  u_{12;xx} & 0 \\
2( - 2\eta^2 - u_{12} ) & - 4\eta^3 + u_{12;x} & 0 \\
0 & 0 & -8\eta^3
\end{pmatrix}.
\end{align*}
Taking into account Remark~\ref{remz}, we obtain the
$\gsl(2,\BBC)$-\/valued zero\/-\/curvature
representation~\eqref{zcrBVV}
for the KdV equation~\eqref{kdv} by omitting the summands $\eta\otimes
Z dx$ and $-4\eta^3\otimes Z dt$ in $A_{\text{KdV}}$ and
$B_{\text{KdV}}$ and by denoting $\eta^2 = \lambda$.

\begin{proposition}\label{propSKdVcov}
The $N{=}2$ supersymmetric $a{=}4$ Korteweg\/--\/de Vries
equation~\eqref{SKdVComponents} admits the covering that, under the
reduction $u_0=u_1=u_2=0$ of~\eqref{SKdVComponents} to the KdV
equation~\eqref{kdv}, reduces to the covering that contains, in the
form of~\eqref{miuracov}, the known Gard\-ner deformation of the KdV
equation~\eqref{kdv}.
\end{proposition}

\begin{proof}
Let us extend the gauge transformation~\eqref{gaugeMiura2BVV}, which
was determined by the element $S$ of the Lie group $SL(2,\BBC)$.
We let
\begin{equation}\label{superBVV2Miura}
S^{N=2}=\begin{pmatrix}
  -1 & -\frac12\veps^{-1} & 0\\
  0 & \veps & 0 \\
  0 & 0 & -\veps
\end{pmatrix}.
\end{equation}
Acting by gauge transformation
~\eqref{superBVV2Miura}
on zero\/-\/curvature representation~\eqref{daszcr}, we obtain the
graded zero\/-\/curvature representation that contains the ``small''
zero\/-\/curvature representation which is originates
from~\eqref{miuracov} and gauge\/-\/equivalent to~\eqref{zcrBVV} for
the KdV equation~\eqref{kdv}. Specifically, we have that
\begin{subequations}\label{supergardnerzcr}
\begin{align}
A = &\begin{pmatrix}
  iu_0 &
    \veps^{-1}(u_0^2 + u_{12}) - i\veps^{-2}u_0 &
      \veps^{-1}(u_{2} - iu_{1})\\
 -\veps &
   iu_0 - \veps^{-1} &
     0 \\
 0 &
   u_2 + iu_1 &
     2iu_0 - \veps^{-1}
\end{pmatrix},\\
B = &\begin{pmatrix}
  b_{11} & b_{12} & b_{13} \\
  b_{21} & b_{22} & b_{23} \\
  b_{31} & b_{32} & b_{33}
\end{pmatrix},
\end{align}
\end{subequations}
where the elements of the matrix  $B$ are as follows
\begin{align*}
b_{11} ={} & 4iu_{0}^3 - 6iu_{0} u_{12}  + 4u_{0} u_{0;x} - iu_{0;xx}
- u_{12;x} - 4iu_{2}u_{1} + \veps^{-1} (2u_{0}^2 - u_{12}  - iu_{0;x})
- i\veps^{-2}u_{0} ,\\
b_{12}={} & \veps^{-1}(4u_{0}^4 + 2u_{0}^2u_{12}  + 4u_{0} u_{0;xx} -
2u_{12}^2 + 4u_{0;x}^2 - u_{12;xx} + u_{2}u_{2;x} + 8u_{2}u_{1}u_{0}
+ u_{1}u_{1;x}) + {}\\
{}&{}+ \veps^{-2}(2iu_{0}^3 - 4iu_{0} u_{12}  + 4u_{0}
u_{0;x} - iu_{0;xx} - u_{12;x} - 2iu_{2}u_{1}) + \veps^{-3} (u_{0}^2 -
u_{12}  - iu_{0;x})-{}\\
{}&{} - i\veps^{-4}u_{0} ,\\
b_{13} = {} & \veps^{-1}( - 5iu_{0} u_{2;x} - 5u_{0} u_{1;x} -
u_{2;xx} + iu_{1;xx} + 8u_{2}u_{0}^2 - 2u_{2}u_{12}  - 4iu_{2}u_{0;x}
- 8iu_{1}u_{0}^2 + {}\\
{}&{}  + 2iu_{1}u_{12}  - 4u_{1}u_{0;x}) + \veps^{-2} ( - u_{2;x} +
iu_{1;x} - 3iu_{2}u_{0}  - 3u_{1}u_{0} ) + \veps^{-3}( - u_{2} +
iu_{1}),\\
b_{21} = {} & 2\veps( - 2u_{0}^2 + u_{12} ) + 2i u_{0}  +
\veps^{-1},\\
b_{22} = {} & 4iu_{0}^3 - 6iu_{0} u_{12}  - 4u_{0} u_{0;x} -
iu_{0;xx} + u_{12;x} - 4iu_{2}u_{1} + \veps^{-1}( - 2u_{0}^2 + u_{12}
+ iu_{0;x}) + {} \\
{}&{}+i\veps^{-1} u_{0}  + \veps^{-3},\\
b_{23} = {} &  u_{2;x} - iu_{1;x} + 4iu_{2}u_{0}  + 4u_{1}u_{0}  +
\veps^{-1}(u_{2} - iu_{1}),\\
b_{31} = {} & \veps ( - u_{2;x} - iu_{1;x} + 4iu_{2}u_{0}  -
4u_{1}u_{0} ) + u_{2} + iu_{1}, \\
b_{32} = {} & 5iu_{0} u_{2;x} - 5u_{0} u_{1;x} - u_{2;xx} - iu_{1;xx}
+ 8u_{2}u_{0}^2 - 2u_{2}u_{12}  + 4iu_{2}u_{0;x} + 8iu_{1}u_{0}^2 -
2iu_{1}u_{12}  - {} \\
{}&{} - 4u_{1}u_{0;x} + \veps^{-1} u_{0} (iu_{2} - u_{1}),\\
b_{33}={} & 2(4iu_{0}^3 - 6iu_{0} u_{12}  - iu_{0;xx} - 4iu_{2}u_{1}) + \veps^{-3}.
\end{align*}

The projective substitution~\eqref{projsub} yields the
two\/-\/dimensional covering over the
$N{=}2$, $a{=}4$ SKdV equation. Under the reduction $u_0=u_1=u_2=0$,
the covering contains~\eqref{miuracov}, which is equivalent to
Gardner's deformation~\eqref{kdvdef} of the KdV equation~\eqref{kdv}.
The $x$-\/components of the derivation 
rules for the nonlocalites $w$ and $f$ are
\begin{align*}
w_x = {} &  \underline{- \veps w^2 + \veps^{-1} (w - u_{12} )}  -
fu_{2} - ifu_{1}  - \veps^{-1}u_{0}^2 - \veps^{-2}iu_{0},\\
f_x = {} &  - \veps wf - iu_{0} f + \veps^{-1} (f - u_{2} + iu_{1}),\\
\intertext{here and in what follows we underline the
covering~\eqref{miuracov} that encodes the ``small''
Gardner deformation for the KdV equation.
The $t$-\/components of the ``large''
covering over the $N{=}2$, $a{=}4$ SKdV are}
w_t = {} & \veps( - 4w^2u_{0}^2 \underline{+ 2w^2u_{12}}  - fwu_{2;x} - ifwu_{1;x}
+ 4ifu_{2}wu_{0}  - 4fu_{1}wu_{0} ) + 2iw^2u_{0}  + {} \\
{}&{}+ 8wu_{0} u_{0;x} \underline{- 2wu_{12;x}} - 5ifu_{0} u_{2;x} + 5fu_{0}
u_{1;x} + fu_{2;xx} + ifu_{1;xx} + fu_{2}w - 8fu_{2}u_{0}^2 + {}\\
{}&{}+ 2fu_{2}u_{12}  - 4ifu_{2}u_{0;x} + ifu_{1}w - 8ifu_{1}u_{0}^2 +
2ifu_{1}u_{12}  + 4fu_{1}u_{0;x} + \veps^{-1}(\underline{w^2} + 4wu_{0}^2 - {} \\
{}&{}\underline{- 2wu_{12}}  - 2iwu_{0;x} - 4u_{0}^4 - 2u_{0}^2u_{12}  - 4u_{0}
u_{0;xx} \underline{+ 2u_{12}^2} - 4u_{0;x}^2 \underline{+ u_{12;xx}} - ifu_{2}u_{0} + {}\\
{}&{} + fu_{1}u_{0}  - u_{2}u_{2;x} - 8u_{2}u_{1}u_{0}  -
u_{1}u_{1;x}) + \veps^{-2}( - 2iwu_{0}  - 2iu_{0}^3 + 4iu_{0} u_{12}
- 4u_{0} u_{0;x} + {} \\
{}&{}+iu_{0;xx} \underline{+ u_{12;x}} + 2iu_{2}u_{1}) + \veps^{-3} (\underline{ - w} -
u_{0}^2 \underline{+ u_{12}}  + iu_{0;x}) + \veps^{-4}iu_{0} ,\\
f_t = {} & 2\veps w( - 2fu_{0}^2 + fu_{12} ) + ( - wu_{2;x} +
iwu_{1;x} + 2ifwu_{0}  - 4ifu_{0}^3 + 6ifu_{0} u_{12}  + 4fu_{0}
u_{0;x} +{} \\
{}&{} + ifu_{0;xx} - fu_{12;x} + 4ifu_{2}u_{1} - 4iu_{2}wu_{0}  -
4u_{1}wu_{0} ) + \veps^{-1}(5iu_{0} u_{2;x} + 5u_{0} u_{1;x} +
u_{2;xx} - {}\\
{}&{}- iu_{1;xx} + fw + 2fu_{0}^2 - fu_{12}  - ifu_{0;x} - u_{2}w -
8u_{2}u_{0}^2 + 2u_{2}u_{12}  + 4iu_{2}u_{0;x} + iu_{1}w +{}\\
{}&{}+ 8iu_{1}u_{0}^2 - 2iu_{1}u_{12}  + 4u_{1}u_{0;x}) +
\veps^{-2}(u_{2;x} - iu_{1;x} - ifu_{0}  + 3iu_{2}u_{0}  + 3u_{1}u_{0}
) + {}\\
{}&{}+ \veps^{-3} ( - f + u_{2} - iu_{1}).
\end{align*}
This proves our claim.
\end{proof}

We finally remark that the reduction $u_0=0$, $u_1 = 0$ (and the
change of notation $u_2 \to \xi$, $u_{12} \to u$) maps this covering
over the $N{=}2$, $a{=}4$ SKdV equation to the covering
which was constructed in Example~\ref{n1example} for
the $N{=}1$ supersymmetric Korteweg\/--\/de Vries equation~\eqref{sKdV}.

\section*{Conclusion}
\noindent%
By now the Gardner deformation problem for the $N{=}2$ supersymmetric
$a{=}4$ Korteweg\/--\/de Vries equation (\cite{MathieuOpen}) is solved.
In this paper we have found the solution which is an alternative to our
previous result in~\cite{HKKW}. Namely, we introduced the nonlocal
bosonic and fermionic variables in such a way that the rules to
differentiate them are consistent by virtue of the super\/-\/equation
at hand and second,
the entire system retracts to the standard KdV
equation and the classical Gardner deformation for it (\cite{Miura68})
under setting to zero the fermionic nonlocal variable and
the first three components of the $N{=}2$ superfield in~\eqref{SKdV}.
At the same time, the nonlocal structure under
study is equivalent to the $\gsl(2\mid 1)$-\/valued zero\/-\/curvature
representation for this super\/-\/equation;
the zero\/-\/curvature representation contains the
non\/-\/removable spectral parameter, which manifest the integrability.


Our new solution of P.~Mathieu's Open problem~2
(see~\cite{MathieuOpen})
relies on the interpretation of both Gardner's deformations
and zero\/-\/curvature representations in similar terms, as a specific
type of nonlocal structures over the equation of motion. However, we
emphasize that generally
there is no one\/-\/to\/-\/one correspondence between the two
constructions, so that the interpretation of deformations in the
Lie\/-\/algebraic language is not always possible. Because this
correlation
between the two approaches to the 
integrability was
not revealed in the canonical formulation of the 
deformation
problem~\cite{MathieuOpen}, there appeared some attempts to solve it
within the classical scheme but the progress was partial.
Still, the use of zero\/-\/curvature representations in this
context could have given the sought deformation long~ago.

Let us also notice that projective substitution~\eqref{projsub}
correlates the super\/-\/dimension of the Lie algebra in a
zero\/-\/curvature representation for a differential equation with the
numbers of bosonic and fermionic nonlocalities over the same system:
a subalgebra of $\mathfrak{gl}(p\mid q)$ yields at most $p-1$ bosonic
and $q$~fermionic variables. This implies that, for a covering
over the $N{=}2$ supersymmetric KdV equation~\eqref{SKdV}
to extend the Gardner deformation~\eqref{kdvdef} in its classical sense
$\mathfrak{m}_\veps\colon\cE_\veps\to\cE$
(see~\cite{Miura68,MathieuNew,TMPh2006}), the extension~$\cE_\veps$
must be the system of evolution equations upon two bosonic and two
fermionic fields. Therefore, one may have to use the
$\mathfrak{sl}(3\mid 2)$-\/valued zero\/-\/curvature representations.
This outlines the working approach to a yet another method of solving
the Gardner deformation problem for the $N{=}2$ supersymmetric
Korteweg\/--\/de Vries systems~\eqref{SKdV}, which we leave as a new
open problem.

\subsubsection*{Acknowledgements}
The authors thank J.~W.~van de Leur, P.~Mathieu,
and M.~A.~Ne\-ste\-ren\-ko
for fruitful discussions and constructive criticisms. This research was
done in part while the first author was visiting at the MPIM (Bonn) and
the second author was visiting at Utrecht University; the hospitality
and support of both institutions is gratefully acknowledged. The
research of the first author was partially supported by NWO
grants~B61--609 and VENI 639.031.623 and performed at Utrecht
University.

\end{document}